\documentclass[11pt]{article}
\usepackage{amsmath,amsfonts,amsthm,amssymb}
\usepackage{algorithm}
\usepackage{graphics,color,url}
\usepackage{graphicx}
\usepackage{epsfig}
\usepackage{fullpage}
\usepackage{color}

\newtheorem{theorem}{Theorem}[section]
\newtheorem{lemma}{Lemma}[section]
\newtheorem{claim}{Claim}[section]

\newtheorem{proposition}{Proposition}[section]

\newcommand{\calK}{{\cal K}}

\newcommand{\calR}{{\cal R}}

\newcommand{\vecw}{{\mathbf w}}
\newcommand{\I}{{\mathcal I}}

\title{Budget Feasible Mechanism Design via Random Sampling \\[.2in]}

\author{
Xiaohui Bei\thanks{Tsinghua University, China. Email: {\tt beixiaohui@gmail.com}.}\\
\and
Ning Chen\thanks{Division of Mathematical Sciences, School of Physical and Mathematical Sciences, Nanyang Technological University, Singapore. Email: {\tt ningc@ntu.edu.sg, ngravin@pmail.ntu.edu.sg}. } \\
\and
Nick Gravin$^{\dag}$\\
\and Pinyan Lu\thanks{Microsoft Research Asia. Email: {\tt pinyanl@microsoft.com}.}
}\date{}

\begin{document}

\maketitle \thispagestyle{empty}

\begin{abstract}
Budget feasible mechanism considers algorithmic mechanism design
questions where there is a budget constraint on the total payment of
the mechanism. An important question in the field is that under
which valuation domains there exist budget feasible mechanisms that
admit `small' approximations (compared to a socially optimal
solution). Singer~\cite{PS10} showed that additive and submodular
functions admit a constant approximation mechanism. Recently,
Dobzinski, Papadimitriou, and Singer~\cite{DPS11} gave an
$O(\log^2n)$ approximation mechanism for subadditive functions and
remarked that: {\em ``A fundamental question is whether, regardless
of computational constraints, a constant-factor budget feasible
mechanism exists for subadditive
function."}\\[-.1in]

In this paper, we give the first attempt to this question. We give a
polynomial time $O(\frac{\log n}{\log\log n})$ sub-logarithmic
approximation ratio mechanism for subadditive functions, improving
the best known ratio $O(\log^2 n)$. Further, we connect budget
feasible mechanism design to the concept of approximate core in
cooperative game theory, and show that there is a mechanism for
subadditive functions whose approximation is, via a characterization
of the integrality gap of a linear program, linear to the largest
value to which an approximate core exists. Our result implies in
particular that the class of XOS functions, which is a superclass of
submodular functions, admits a constant approximation mechanism. We
believe that our work could be a solid step towards solving the
above fundamental problem eventually, and possibly, with an
affirmative answer.
\end{abstract}

\newpage

\section{Introduction}

Consider a scenario where a company is running a set of machines and each of which serves a set of jobs. There is an incurred expense for each machine to serve the jobs, and the total expense of the company is the sum of the expenses of all machines. Assume now the company would like to save its running expense by removing some of served jobs and paying those jobs a certain amount of subsidy. We may assume that every job has a cost of being not served (or equivalently, benefit of being served); the bottom line is therefore to have their cost compensated.
The question that the company considers is that which jobs should be chosen such that the saved expense as much as possible given a universal budget constraint.

Formally, there is a set of agents (i.e., jobs) $A$, and for any subset $S\subseteq A$ there is a public known valuation $v(S)$. (In the above example, $v(S)$ gives how much expenses it can be saved if $S$ is removed.) Each agent $i\in A$ has a cost $c(i)$, which gives an incurred cost to the agent if he is selected.
This defines a natural optimization problem, i.e., find a subset $S$ that maximizes $v(S)$ subject to $\sum_{i\in S}c(i)\le B$, where $B$ is a sharp budget which gives an upper bound of compensation that can be distributed among agents. The budgeted optimization problem has been considered in a variety of domains with respect to different valuation functions, e.g., additive (i.e., knapsack), submodular, and subadditive.

However, agents, as self-interested entities, may want to get as much subsidy as possible. In particular, they can hide their true incurred cost $c(i)$ (which is known only by themselves) and claim `any' amount, say $b(i)$.
We therefore adopt the approach of mechanism design to manage self-interested, but strategic, behaviors of the agents.
Specifically, given submitted bids $b(i)$ from all agents, a mechanism decides a winner set $S$ and a payment $p(i)$ to each winner $i$.
A mechanism is called {\em truthful} (a.k.a. incentive compatible) if it is a dominant strategy for every agent to submit his true cost, i.e., $b(i)=c(i)$.
Truthfulness is one of the central solution concepts in mechanism design; it ensures that every participant will behave precisely according to the mechanism protocol and his true interest.

Our problem has an important and practical extra ingredient: Budget, i.e., the total payment of a mechanism should be upper bounded by $B$.
The budget constraint introduces a new dimension to mechanism design and restricts the space of truthful mechanisms.
For example, in a single parameter domain, a monotone allocation rule plus its associated threshold payment, while still gives a sufficient and necessary condition for truthfulness~\cite{myerson}, may not necessarily generate a budget feasible solution.
Thus, a number of well known truthful designs (e.g., the seminal VCG mechanism~\cite{vickrey,clarke,groves}) do not apply, and new ideas have to be developed.

Another significant change due to the budget constraint is that, unlike the VCG mechanism which always generates a socially optimal solution, we cannot hope to have an output which is both socially optimal and budget feasible even if we are given unlimited computational power. Indeed, in a simple example like path procurement (whose valuation $v(\cdot)$ is a superadditive function), any budget feasible mechanism can have an arbitrarily bad solution. Therefore, the question that one may ask is that
under which valuation domains there exist budget feasible truthful mechanisms that admit `small' approximations (compared to a socially optimal solution).

The answer to the question crucially depends on the properties and classifications of the considered valuation function.
In particular, given the following hierarchy for the functions~\cite{LLN01}:
\[\textup{additive $\subset$ gross substitutes $\subset$ submodular $\subset$ XOS $\subset$ subadditive,}\]
which one admits a positive answer?
Singer~\cite{PS10} initiated the study of approximate budget feasible mechanism design and gave constant approximation mechanisms for additive and submodular functions.
In recent work, Dobzinski, Papadimitriou, and Singer~\cite{DPS11} considered subadditive functions and showed an $O(\log^2n)$ approximation mechanism.
Further, it was remarked that:
\begin{quote}
{\em ``A fundamental question is whether, regardless of computational constraints, a constant-factor budget feasible mechanism exists for subadditive function."}\\[-.3in]
\begin{flushright}
--- Dobzinski, Papadimitriou, Singer~\cite{DPS11}
\end{flushright}
\end{quote}

\paragraph{Our Results.} In this paper, we give the first attempt to this question. Our first
result is a sub-logarithmic approximation ratio mechanism for
subadditive functions, improving the best known ratio $O(\log^2 n)$.

\medskip
\noindent \textbf{Theorem 1.} \textit{There is a polynomial time budget feasible truthful mechanism for subadditive functions with an approximation ratio of $O(\frac{\log n}{\log\log n})$, where $n$ is the number of agents.}

\medskip
Here we assume that we are given an demand oracle for the subadditive valuation function. This is in the same setting
as the mechanism in \cite{DPS11} since it was proved that a value oracle is not sufficient~\cite{PS10}.
The sub-logarithmic approximation further sheds light on the hope of a positive answer to the above question.
We continue to explore budget feasible mechanisms under the domain of XOS and subadditive functions.
Consider the following linear program (LP), where $\alpha(\cdot)$'s are variables.
\begin{eqnarray*}
&\min& \sum_{S\subseteq A} \alpha(S)\cdot v(S)\\
&s.t.& \alpha(S)\ge 0,\quad\quad \forall\ S\subseteq A \\
&& \sum_{S:\ i\in S} \alpha(S) \ge 1,\quad \forall\ i\in A
\end{eqnarray*}
In the above LP, if we consider each $\alpha(S)$ as the fraction covered by the subset $S$, the last constraint requires that all items in $A$ are fractionally covered;
hence, it describes a linear program for the set cover of $A$. An important observation of the LP is that for any subadditive function $v(\cdot)$, the value of the optimal integral solution is precisely $v(A)$.

The above LP has a strong connection to cores of cost sharing games
(considering $v(\cdot)$ instead as a cost function), which is a
central notion in cooperative game theory~\cite{agt-book}. Roughly
speaking, the core of a game is a stable cooperation among all
agents to share $v(A)$ where no subset of agents can benefit by
breaking away from the grand coalition. It is well known that the
cores of many cost sharing games are empty. This motivates the notion
of $\alpha$-approximate core, which requires all the agents to share
only $\alpha$ fraction of $v(A)$. The classic Bondareva-Shapley
Theorem~\cite{bondareva63,shapley67} says that for subadditive
functions, the largest value $\alpha$ for which the
$\alpha$-approximate core is nonempty is equal to the integrality
gap of the LP. Further, the integrality gap of the LP is one (i.e.,
$v(A)$ is also an optimal fractional solution) if and only if the
valuation function is XOS; this is also equivalent to the
non-emptiness of the core.

Given an instance of our problem with agents set $A$, we may
consider the above LP and its integrality gap for every subinstance
$A'\subseteq A$; among which the largest integrality gap
characterizes the worst scenario between the optimal integral and
fractional solution of the problem. Our second result is the
following.

\medskip
\noindent \textbf{Theorem 2.} \textit{There is a budget feasible
truthful mechanism for subadditive functions with approximation
ratio linear to the largest integrality gap over all subinstances of
the above LP. In particular, for XOS functions, the mechanism has a
constant approximation ratio.}

\medskip

For some special subadditive functions whose integrality gaps are
bounded by constants (XOS is one such example), our mechanism will
have a constant approximation. Note that the mechanism may have
exponential running time. For some special XOS functions like matching
and clique, the mechanism can be implemented in polynomial time
(given a demand query to the valuation function, see
Section~\ref{sec:special}). Further, our mechanisms also work for
non-monotone functions and a generalized subadditive functions (see
Section~\ref{section-extension}).

Our results show an interesting connection between budget feasible
mechanism design and integrality gap, as well as the existence of
$\alpha$-approximate core. While our mechanisms do not answer the
above fundamental question posed in~\cite{DPS11} directly, we
believe that they could be a solid step towards solving the problem
eventually, and possibly, with an affirmative answer.

In the design of budget feasible mechanisms, due to the sharp budget
constraint, the major approach used by previous works,
e.g.,~\cite{PS10,CGL11,DPS11}, is based on a simple idea of adding
agents one by one greedily and carefully to ensure that the budget
constraint is not violated. Our mechanisms use another simple, but
powerful, approach: Random sampling. We add agents into a test set
$T$ with probability half each and compute an optimal budgeted
solution on $T$ (to derive efficient computation, a constant
approximation to the optimum suffices). Note that all agents in $T$
are only for the purpose of `evaluation' and will not be winners
anyway. The computed optimal solution on $T$ gives a close estimate
to the optimal solution of the whole set with a high probability. We
then, using the evaluation from $T$ as a threshold, compute a real
winner set from the remaining agents. Random sampling appears as a
powerful approach and has been used in other domains of mechanism
design, e.g., digital goods auctions~\cite{GHK}, secretary
problem~\cite{BIK,BDG}, social welfare maximization~\cite{Dob07},
and mechanism design without money~\cite{CGL}. It is intriguing to
find applications of random sampling in other mechanism design
problems.

\paragraph{Related Work.} Our work falls into the subject of
algorithmic mechanism design, which is a fascinating field initiated
by the seminal work of Nisan and Ronen~\cite{NR99}. There are two
main threads in algorithmic mechanism design: approximate social
welfare with efficient computation or with frugal payment; our work
belongs to the latter.

As mentioned earlier, the study of approximate mechanism design with
a budget constraint was originated by Singer~\cite{PS10} and constant
approximation mechanisms were given for additive and submodular
functions. The approximation ratios were later improved
in~\cite{CGL11}. Dobzinski, Papadimitriou, and Singer~\cite{DPS11}
considered subadditive functions and showed an $O(\log^2n)$
approximation mechanism. They also considered cut function, which is
a special non-monotone function, and gave constant approximation
mechanisms. Ghosh and Roth~\cite{GR} considered a budget feasible
mechanism design model for selling privacy where there are
externalities for each agent's cost.

In an independent work, Badanidiyuru, Dobzinski, and Oren~\cite{BDO11}
considered maximizing social welfare for subadditive functions with a knapsack constraint and gave a $2+\epsilon$ approximation algorithm
with demand queries. We consider the same problem and present a $4+\epsilon$ approximation algorithm; the algorithm is used as a subroutine in our truthful mechanisms.
However, the focus of our paper is completely different from~\cite{BDO11}: ours is on truthful mechanism design, whereas~\cite{BDO11} is on social welfare maximization.

\section{Preliminaries}\label{sec:preliminaries}

In a marketplace there are $n$ agents (or items), denoted by $A$. Each agent $i\in A$ has a privately known incurred {\em cost} $c(i)\ge 0$. For any given subset $S\subseteq A$ of agents, there is a publicly known valuation $v(S)$, meaning the social welfare derived from $S$. We assume that $v(\emptyset)=0$ and the valuation function is monotone, i.e., $v(S)\le v(T)$ for any $S\subset T\subseteq A$ (in Section~\ref{section-extension}, we will discuss how to remove the monotone assumption).

We will consider XOS and subadditive functions in the paper; both are rather general classes and contain a number
of well studied functions as special cases, e.g., additive, gross substitutes, and submodular.
\begin{itemize}
\item Subadditive (a.k.a. complement free): $v(S)+v(T)\ge v(S\cup T)$ for any $S,T\subseteq A$.
\item XOS (a.k.a. fractionally subadditive): There is a set of linear functions $f_1,\ldots,f_m$ such that $v(S)=\max\big\{f_1(S),f_2(S),\ldots,f_m(S)\big\}$ for any $S\subseteq A$. Note that the number of functions $m$ can be exponential in $n=|A|$.

    Another definition is that $v(S)\le \sum_{T\in 2^A}x(T)\cdot v(T)$ whenever $\sum_{T: i\in T}x(T)\ge 1$ for any $i\in S$, where $0\le x(T)\le 1$ and $2^A$ is the power set of $A$. That is, if every element in $S$ is fractionally covered, then the sum of the values of all subsets weighted by the corresponding coefficients is at least as large as $v(S)$. Feige~\cite{feige09} showed that the two definitions are equivalent.
\end{itemize}

Our objective is to pick a subset of agents with maximum possible valuation given a sharp budget $B$ to cover their
incurred costs, i.e., $\max_{_{S\subseteq A}}v(S)$ given $c(S) = \sum_{i\in S}c(i) \le B$.
However, agents, as self-interested entities, have their own objective as well; each agent $i$ may not tell his true
privately known cost $c(i)$, but, instead, submit a {\em bid} $b(i)$ strategically.
We use mechanism design and its solution concept truthfulness to manage strategic behaviors of the agents.
Upon receiving $b(i)$ from each agent, a mechanism decides an {\em allocation} $S\subseteq A$ of the winners and a {\em
payment} $p(i)$ to each $i\in A$. We assume that the mechanism has no positive transfer (i.e., $p(i)=0$ if $i\notin S$)
and is individually rational (i.e., $p(i)\ge b(i)$ if $i\in S$).

In a mechanism, agents bid strategically to maximize their utilities, which is $p(i)-c(i)$ if $i$ is a winner and $0$
otherwise. We say a mechanism is {\em truthful} if it is of the best interest for each agent to report his true cost,
i.e., $b(i)=c(i)$. For randomized mechanisms, we consider universal truthfulness in this paper: a randomized mechanism is
called {\em universally truthful} if it takes a distribution over deterministic truthful mechanisms.

Our setting is a single parameter domain, as each agent has only one private parameter which is his cost. It is well-
known~\cite{myerson} that in the single parameter setting, a mechanism is truthful if and only if its allocation rule is
monotone (i.e., a winner keeps winning if he unilaterally decreases his bid) and the payment to each winner is his
threshold bid (i.e., the maximal bid for which the agent still wins). Therefore, to the end of designing a truthful
mechanism, it suffices to design a monotone allocation.

A mechanism is said to be {\em budget feasible} if its total payment is within the budget constraint, i.e., $\sum_i p(i)
\le B$. Assume without loss of generality that $c(i)\le B$ for any agent $i\in A$, since such agent will never win in any
budget feasible truthful mechanism.
We evaluate a mechanism according to its {\em approximation ratio}, which is defined as $\max_{_I} \frac{opt(I)}{\mathcal{M}(I)}$, where $\mathcal{M}(I)$ is the (expected) value of a mechanism $\mathcal{M}$ on instance $I$ and $opt(I)$ is the optimal value of the following problem: $\max_{_{S\subseteq A}}v(S)$ subjected to $c(S)\le B$.
Our goal in the present paper is to design truthful budget feasible mechanisms for XOS and subadditive functions with small approximation ratios.

\section{A Sub-Logarithmic Approximation Mechanism}

\newcommand{\paymentsharing}{{\sc SA-random-sample}}
\newcommand{\submain}{{\sc SA-mechanism-main}}
\newcommand{\litem}{{\sc mechanism-largest-item}}
\newcommand{\submax}{{\sc SA-alg-max}}

In this section we give an $o(\log n)$ approximation truthful mechanism for subadditive valuation function.
Note that the representation of a subadditive function usually requires exponential size in $n$.
Thus, we assume that we are given access to a {\em demand oracle}, which, for any given price vector $p(1),\ldots,p(n)$,
returns us a subset $T\in \arg\max_{S\subseteq A}\left(v(S)-\sum_{i\in S}p(i)\right)$.
A demand oracle enables us to evaluate the values of the function $v(\cdot)$,
and a polynomial number of queries can be asked in a polynomial time mechanism.

\subsection{Subadditive Function Maximization with Budget}

We first describe an algorithm that approximates $\max_{S\subseteq A}v(S)$ given that $c(S)\le B$.
That is, we ignore for a while strategic behavior of agents and consider a pure maximization problem
where the objective is to pick a subset with maximum possible valuation under the budget constraint.
Dobzinski et al.~\cite{DPS11} considered the same question and gave a 4 approximation algorithm for
the unweighted case (i.e., the restriction is on the size of selected subset).
Our algorithm extends their result to the weighted case and runs in polynomial
in $n$ time if we are given a demand oracle.

\begin{center}
\small{}\tt{} \fbox{
\parbox{6.0in}{\hspace{0.05in} \\[-0.05in] \underline{\submax}
\begin{itemize}
\item Let $v^*=\max_{i\in A} v(i)$ and  ${\cal V}=\{v^*,2 v^*,\ldots, n v^*\}$
\item For each $v\in \cal{V}$
    \begin{itemize}
    \item Set $p(i)=\frac{v}{2B}\cdot c(i)$ for each $i\in A$, and find $T\in \arg\max_{S\subseteq A}\left(v(S)-\sum_{i\in S}p(i)\right)$.
    \item Let $S_v=\emptyset$.
    \item If $v(T)< \frac{v}{2}$, then continue to next $v$.
    \item Else, in decreasing order of $c(i)$ put items from $T$ into $S_v$ while budget constraint is not violated.
    \end{itemize}
\item Output: $S_v$ with the largest value $v(S_v)$ for all $v\in {\cal V}$.
\end{itemize}
}}
\end{center}

\begin{lemma}\label{lemma-SA-max}
\submax\ is an $8$ approximation algorithm for subadditive function maximization given a demand oracle.
\end{lemma}
\begin{proof}
Let $S^*$ be an optimal solution. Note that $v(S^*)\ge v^*=\max_{i\in A} v(i)$ and $c(S^*)\le B$.
For all $v \le v(S^*)$, we first prove that the algorithm will generate a non-empty set $S_v$ with $v(S_v)\ge \frac{v}{4}$.
Since $T$ is the maximum set returned by the oracle, we have
\[ v(T)-\frac{v}{2B} c(T) \geq  v(S^*)-\frac{v}{2B} c(S^*) \geq  v-\frac{v}{2B} \cdot B \geq  \frac{v}{2}\]
Hence, $v(T)\geq \frac{v}{2}$.
If $c(T)\le B$, then $S_v=T$ and we are done. Otherwise,
by the greedy procedure of picking items from $T$ to $S_v$, we are guaranteed that $c(S_v)\geq \frac{B}{2}$.
Assume for contradiction that $v(S_v)<\frac{v}{4}$.
Then
\begin{eqnarray*}
v(T\setminus S_v)-\frac{v}{2B} c(T\setminus S_v)
&\geq & v(T)-v(S_v)-\frac{v}{2B} \big(c(T) - c(S_v)\big)\\
&> & v(T)-\frac{v}{4}- \frac{v}{2B} c(T) + \frac{v}{2B} \cdot \frac{B}{2} \\
&= & v(T)-\frac{v}{2B} c(T)
\end{eqnarray*}
The later contradicts to the definition of $T$, since $T\setminus S_v$ is then better than $T$.
Thus, we always have $v(S_v)\geq \frac{v}{4}$ for each $v \le v(S^*)$.
Since the algorithm tries all possible $v\in \cal{V}$ (including one with $\frac{v(S^*)}{2}<v\leq v(S^*)$) and outputs
the largest $v(S_v)$, the output is guaranteed to be within a factor of 8 to the optimal value $v(S^*)$.
\end{proof}

Note that we can actually modify the algorithm to get a $4+\epsilon$
approximation with runtime polynomial in $n$ and $\frac{1}{\epsilon}$. To do so one
may simply replace ${\cal V}$ by a larger set $\big\{\epsilon v^*,2\epsilon v^*,\ldots, \lceil\frac{n}{\epsilon}\rceil
\epsilon v^*\big\}$. Both algorithms suffice for our purpose; for the rest of the paper, for simplicity we will use the 8 approximation
algorithm to avoid extra parameter $\epsilon$ in the analysis.

We will use \submax\ as a subroutine to build a mechanism for subadditive functions in the subsequent section.
When there are different possible sets maximizing $v(S)-\sum_{i\in S}p(i)$,
we require the algorithm to compute a fixed set $T$ (i.e., the result will be the same for all possible
answers on oracle queries).
This property is important for truthfulness of our mechanism.
To implement this, we assume that there is a fixed order of all items $i_1\prec i_2\prec \cdots \prec i_n$.
We first compute
\[T_1\in \arg\max_{S\subseteq A}\bigg(v(S)-\sum_{i\in S}p(i)\bigg) \ \ \textup{and} \ \ T_2\in \arg\max_{S\subseteq
A\setminus \{i_1\}}\bigg(v(S)-\sum_{i\in S}p(i)\bigg) .\]
If $v(T_1)-\sum_{i\in T_1}p(i) = v(T_2)-\sum_{i\in T_2}p(i)$, we know that there is a subset without $i_1$ that gives
us the maximum; thus, we can ignore $i_1$ for consideration.
If $v(T_1)-\sum_{i\in T_1}p(i) > v(T_2)-\sum_{i\in T_2}p(i)$, we know that $i_1$ should be included in any optimal
solution; hence, we will always include $i_1$ and proceed the process iteratively for $i_{2},i_{3},\ldots,i_n$. One can
see that this process gives a fixed outcome that maximizes $v(S)-\sum_{i\in S}p(i)$.

\subsection{Mechanism}

In this section, we will describe our mechanism for subadditive functions.

\begin{center}
\small{}\tt{} \fbox{
\parbox{6.3in}{\hspace{0.05in} \\
[-0.05in] \underline{\paymentsharing}
\begin{enumerate}
\item Pick each item independently at random with probability $\frac{1}{2}$ into group $T$.
\item Run \submax\ for items in group $T$; let $v$ be the value of the returned subset.
\item For $k=1$ to $|A\setminus T|$
\begin{itemize}
  \item Run \submax\ on the set $\left\{i\in A\setminus T~|~c(i)\le \frac{B}{k}\right\}$
  where each item has a cost $\frac{B}{k}$, denote the output by $X$.
  \item If $v(X)\ge\frac{\log\log n}{80\log n} \cdot v$
  \begin{itemize}
  \item Output $X$ as the winner set and pay $\frac{B}{k}$ to each item in $X$.
  \item Halt.
  \end{itemize}
\end{itemize}
\item Output $\emptyset$.
\end{enumerate}
} }
\end{center}

In the above mechanism, we first sample in expectation half of items to form a testing group $T$,
and then use \submax\ to compute an approximate solution for items in $T$ given the budget constraint $B$.
As it can be seen in the analysis of the mechanism, the computed value $v$ is in expectation within a constant factor
from the optimal value of the whole set $A$.
That is, we are able to learn the rough value of the optimal solution by random sampling.
Next we consider the remaining items $A\setminus T$ and try to find a subset $X$ with relatively big value in which
every item willing to ``share'' the budget $B$ at a fixed share $\frac{B}{k}$.
(This part of our mechanism can be viewed as a reversion of the classic cost sharing mechanism.)
Finally, we use the information $v$ from random sampling as a benchmark to determine
whether $X$ should be a winner set or not.

The final mechanism for subadditive functions, which we denote by \underline{\submain}, is a uniform
distribution of the above \paymentsharing\ and the following one which simply picks a single item with the largest value.

\begin{center}
\small{}\tt{} \fbox{
\parbox{3.5in}{\hspace{0.05in} \\
[-0.05in] \underline{\litem}
\begin{itemize}
\item Let $i\in \arg\max_i v(i)$ be the winner.
\item Pay all budget $B$ to the winner $i$.\\[-.2in]
\end{itemize}
} }
\end{center}

\begin{theorem}\label{theorem-SA-loglog}
\submain\ runs in polynomial time given a demand oracle and is a universally truthful budget feasible mechanism with an
approximation ratio of $O(\frac{\log n}{\log \log n})$.
\end{theorem}

To the end of proving the claim, we first establish the following
lemma.

\begin{lemma}\label{lem:probability}
For any given subset $S\subseteq A$ and a positive integer $k$,
%
%
assume that $v(S)\ge k\cdot v(i)$ for any $i\in S$. Further, suppose
that $S$ is divided uniformly at random into two groups $T_1$ and
$T_2$. Then, with probability at least $\frac{1}{2}$, we have
$v(T_1)\ge \frac{k-1}{4k}v(S)$ and $v(T_2)\ge \frac{k-1}{4k}v(S)$.
\end{lemma}
\begin{proof}
We first claim that there are disjoint subsets $S_1$ and $S_2$ with $S_1\cup S_2 =S$ such that $v(S_1)\ge\frac{k-1}{2k}v(S)$ and $v(S_2)\ge\frac{k-1}{2k}v(S)$.
This can be seen by the following recursive process: Initially let $S_1=\emptyset$ and $S_2=S$; and we move items from
$S_2$ to $S_1$ arbitrarily until the point when $v(S_1)\ge\frac{k-1}{2k}v(S)$. Consider the $S_1,S_2$ at the end of the process; we claim that at this point, we also have $v(S_2)\ge\frac{k-1}{2k}v(S)$.
Note that $v(S)\le v(S_1)+v(S_2)$.
Let $i$ be the last item moved from $S_2$ to $S_1$; therefore, $v(S_1\setminus \{i\})<\frac{k-1}{2k}v(S)$, which implies that $v(S_2\cup \{i\})>\frac{k+1}{2k}v(S)$.
Thus, $v(S_2)+v(i)\ge v(S_2\cup \{i\})>\frac{k+1}{2k}v(S)$. As $v(i)\le \frac{1}{k}v(S)$, we know that $v(S_2)>\frac{1}{2}v(S)>\frac{k-1}{2k}v(S)$.

Consider sets $X_1=S_1\cap T_1$, $Y_1=S_1\cap T_2$, $X_2=S_2\cap T_1$ and $Y_2=S_2\cap T_2$.
Due to subadditivity we have
$\frac{k-1}{2k}v(S)\le v(S_1)\le v(X_1)+v(Y_1)$; hence,
either $v(X_1)\ge \frac{k-1}{4k}v(S)$ or $v(Y_1)\ge \frac{k-1}{4k}v(S)$. Similarly, we have that
either $v(X_2)\ge \frac{k-1}{4k}v(S)$ or $v(Y_2)\ge \frac{k-1}{4k}v(S)$.
Clearly, partitioning $S_1$ into $X_1$,$Y_1$ and partitioning $S_2$ into $X_2$, $Y_2$ are independent to each other.
Therefore, with probability $\frac{1}{2}$ the most valuable parts of $S_1$'s partition and $S_2$'s partition will get
into different sets $T_1$ and $T_2$, respectively. Thus the lemma follows.
\end{proof}

\begin{proof}[Proof of Theorem~\ref{theorem-SA-loglog}.]
Let $S=A\setminus T$.
It is obvious that the mechanism runs in polynomial time since \submax\ is in polynomial time.
If the mechanism chooses \litem, certainly it is budget feasible as the total payment is precisely $B$.
If it chooses \paymentsharing, either no item is a winner or $X$ is selected as the winner set. Note that $|X|\le k$ and each item in $X$ gets a payment of $\frac{B}{k}$. It is therefore budget feasible as well.

\medskip \noindent
{\em (Truthfulness.)} Truthfulness for \litem\ is obvious (as the outcome is irrelevant to the submitted bids). Next we will prove that \paymentsharing\ is truthful as well.
The random sampling step does not depend on the bids of the items, and items in $T$ have no incentive to lie as they cannot win anyway. Hence, it suffices to only consider items in $S$.
Observe that every agent will be a candidate to the winning set only if $c(i)\le \frac{B}{k}$.
Consider any item $i\in S$ and fixed bids of other items. There are the following three possibilities if $i$ reports his true cost $c(i)$.
\begin{itemize}
  \item Item $i$ wins with a payment $\frac{B}{k}$. Then we have $c(i)\le \frac{B}{k}$ and his utility is $\frac{B}{k}-
  c(i)\geq 0$. If $i$ reports a bid which is still less than or equal to $\frac{B}{k}$, the
  output and all the payments do not change. If $i$ reports a bid which is larger than $\frac{B}{k}$, he still could not
  win for a share larger than $\frac{B}{k}$ and will not be considered for all smaller shares. Therefore, he derives $0$
  utility. Thus for either case, $i$ does not have incentive to lie.
  \item Item $i$ loses and payment to each winner is $\frac{B}{k}\geq c(i)$. In this case, if $i$ reduces or increases
  his bid, he cannot change the output of the mechanism. Thus $i$ always has zero utility.
  \item Item $i$ loses and payment to each winner is $\frac{B}{k} < c(i)$ or the winning set is empty.
   In this case, if $i$ reduces his bid, he will not change the process of the mechanism until the payment offered by the
   mechanism is less than $c(i)$. Thus, even if $i$ could win for some value $k$, the payment he gets would be less than
   $c(i)$, in which case his utility is negative. If $i$ increases his bid, he lose and thus derives
   zero utility.
\end{itemize}
Therefore, \paymentsharing\ is a universally truthful mechanism.

\medskip \noindent
{\em (Approximation Ratio.)}
It remains to estimate the approximation ratio. For any subset $Z\subseteq A$, let $OPT(Z)$ denote the optimal solution
over the agents in $Z$ under the budget constraint; and $OPT=OPT(A)$ denote the optimal solution for the whole set.

If there exists an item $i\in A$ such that $v(i)\geq \frac{1}{2} v(OPT)$, then \litem\ will output an item with value at least $\frac{1}{2} v(OPT)$ and we are done. In the following, we assume that $v(i)<\frac{1}{2} v(OPT)$ for all $i\in A$. Then, by Lemma \ref{lem:probability}, with probability at least $\frac{1}{2}$ we have $v(OPT(T)) \geq \frac{1}{8} v(OPT)$ and $v(OPT(S)) \geq \frac{1}{8} v(OPT)$. Hence, with probability at least $\frac{1}{4}$ we have
\begin{equation}\label{eq:separation}
v(OPT(S)) \geq  v(OPT(T)) \geq\frac{1}{8} v(OPT).
\end{equation}
Therefore, it suffices to prove that the main mechanism has an
approximation ratio of $O(\frac{\log n}{\log \log n})$ given the
inequalities \eqref{eq:separation}.

Since \submax\ is an $8$ approximation of $v(OPT(T))$, we have
$v\geq \frac{1}{8}v(OPT(T)) \geq \frac{1}{64}v(OPT)$. Clearly, if
\paymentsharing\ outputs a non-empty set, then its value is at least
$\frac{\log\log n}{80\log n} \cdot v\geq \frac{\log\log n}{5120 \log
n} \cdot v(OPT)$. Hence, it remains to prove that the mechanism will
always output a non-empty set given formula~(\ref{eq:separation}).

Let $S^*=\{1,2,3,\ldots, m\}\subseteq S$ be an optimal solution of $S$ given the budget constraint $B$ and $c_1\geq c_2 \geq \cdots \geq c_m$. We recursively divide the agents in $S^*$ into different groups as follows:
\begin{itemize}
  \item Let $\alpha_1$ be the largest integer such that $c_1 \leq \frac{B}{\alpha_1}$.
        Put the first $\min\{\alpha_1,m\}$ agents into group $Z_1$.
  \item Let $\beta_r=\alpha_1+\dots+\alpha_r$. If $\beta_r < m$ let $\alpha_{r+1}$ be the largest integer such that $c_{_{\beta_r+1}} \leq \frac{B}{\alpha_{r+1}}$;
        put the next $\min\{\alpha_{r+1},m-\beta_r\}$ agents into group $Z_{r+1}$.
\end{itemize}

Let us denote by $x+1$ the number of groups.
Since items in $S^*$ are ordered by $c_1\geq c_2 \geq \cdots \geq c_m$, we have $\alpha_{r+1}\ge \alpha_r$ for any $r$.
If there exists a set $Z_j$ such that $v(Z_j)\geq\frac{\log\log n}{10\log n}\cdot v$,
then the mechanism does not output an empty set, as it could buy $\alpha_j$ items at price $\frac{B}{\alpha_j}$ given that \submax\ is an
$8$-approximation and the threshold we set is $v(Z_j)\geq\frac{\log\log n}{80\log n}\cdot v$.
Thus, we may assume that $v(Z_j) < \frac{ \log \log n}{10 \log n} \cdot v$ for each $j=1,2,\ldots, x+1$.
On the other hand, by subadditivity, we have
\[\sum_{j=1}^{x+1} v(Z_j) \ge v(S^*) = v(OPT(S)) \geq v(OPT(T))\ge v.\]
Putting the two inequalities together, we can conclude that $(x+1)\cdot \frac{\log\log n}{10\log n}\cdot v > v$, which implies that
\[x > \frac{5\log n}{\log \log n} \geq \frac{5\log m}{\log \log m}. \]
%

On the other hand, since $S^*=\{1,2,3,\ldots, m\}$ is a solution for $S$ within the budget constraint, we have that $\sum_{j=1}^{m} c_j \leq B$. Further, since
$c_1 > \frac{B}{\alpha_1+1}, c_{_{\beta_1+1}} > \frac{B}{\alpha_2+1}, \ldots, c_{_{\beta_x+1}} > \frac{B}{\alpha_{x+1}+1}$, we have
\[
B \geq \sum_{j=1}^{m} c_j \geq c_1 + \alpha_1 c_{_{\beta_1+1}} +\cdots +\alpha_x c_{_{\beta_x+1}} >
\frac{B}{\alpha_1+1} + \frac{\alpha_1 B}{\alpha_2+1} +\cdots +\frac{\alpha_x B}{\alpha_{x+1}+1}.
\]
Hence,
\[
1 \geq \frac{1}{\alpha_1+1} + \frac{\alpha_1 }{\alpha_2+1} +\cdots +\frac{\alpha_x }{\alpha_{x+1}+1} \geq
\frac{1}{2\alpha_1} + \frac{\alpha_1 }{2\alpha_2} +\cdots +\frac{\alpha_x }{2\alpha_{x+1}}.
\]
In particular, we get
\[
2 \geq \frac{1}{\alpha_1} + \frac{\alpha_1 }{\alpha_2} +\cdots +\frac{\alpha_{x-1} }{\alpha_x}\geq x\sqrt[^x]{\frac{1}{\alpha_1}\frac{\alpha_1 }{\alpha_2}\cdots\frac{\alpha_{x-1} }{\alpha_x}},
\]
where the last inequality is simply the inequality of arithmetic and geometric means.
Hence, we get $2\ge x\sqrt[^x]{\frac{1}{\alpha_x}}$, which
is equivalent to $\alpha_x\ge (\frac{x}{2})^{x}$. Now plugging in the fact that $m\ge \alpha_x$ and $x\ge\frac{5\log m}{\log \log m}$, we come to a contradiction.
This concludes the proof.
\end{proof}

\section{Integrality-Gap Approximation Mechanisms}

The mechanism \submain\ gives a sub-logarithmic approximation for
subadditive functions. The next question one would ask is whether
there exists a constant approximation truthful mechanism. In this
section we give another mechanism attempting to answer this
question. Our mechanism has an approximation ratio which equals to
the integrality gap of a linear program. For special cases when the
integrality gap can be bounded by a constant (e.g., all XOS
functions have integrality gap $1$), our mechanism will have a
constant approximation ratio.

For simplicity, we will first present our mechanism for XOS
functions. Next in Section~\ref{sec:subadditive}, we will discuss
how to generalize it to subadditive functions.

\subsection{XOS Functions}\label{sec:xos}

\newcommand{\XOSsample}{{\sc XOS-random-sample}}
\newcommand{\XOSmain}{{\sc XOS-mechanism-main}}
\newcommand{\AddM}{{\sc Additive-mechanism}}

We will first consider XOS functions. Given an XOS function $v(\cdot)$, by its definition,
assume that $v(S)=\max\left\{f_1(S),f_2(S),\ldots,f_m(S)\right\}$ for any $S\subseteq A$, where each $f_j(\cdot)$ is a nonnegative additive function,
i.e., $f_j(S)=\sum_{i\in S}f_j(i)$. Note that the value $m$ may not be bounded by a polynomial of $n=|A|$.

In our mechanism, we use a random mechanism \AddM\ for additive valuation functions
as an auxiliary procedure, where \AddM\ is a universally truthful mechanism
and has an approximation factor of at most $3$ (see, e.g., Theorem~B.2, \cite{CGL11}).

\begin{center}
\small{}\tt{} \fbox{
\parbox{6.25in}{\hspace{0.05in} \\
[-0.05in] \underline{\XOSsample}
\begin{enumerate}
\item Pick each item independently at random with probability $\frac{1}{2}$ into group $T$.
\item Compute an optimal solution $OPT(T)$ for items in $T$ given budget $B$.
\item Set a threshold $t=\frac{v(OPT(T))}{8B}.$
\item Consider items in $A\setminus T$ and find a set $S^*\in \arg\max\limits_{S\subseteq A\setminus T}\big\{v(S)-t\cdot c(S)\big\}.$
\item Find an additive function $f$ with $f(S^*)=v(S^*)$ in the XOS representation of $v(\cdot)$.
\item Run \AddM\ for function $f(\cdot)$ with respect to set $S^*$ and budget $B$.
\item Output the result of \AddM. \\[-.2in]
\end{enumerate}
} }
\end{center}

In the above mechanism, again we use the approach of random sampling to evaluate the optimal solution and use this information to compute a proper threshold $t$ for the rest of items.
Specifically, we find a subset $S^*$ with the largest difference between its value and cost, multiplied by the threshold $t$
(in the computation of $S^*$, if there are multiple choices, again, we pick one according to a fixed order).
Finally, we use the property of XOS functions to find a linear representation of $v(S^*)$ and run a truthful mechanism for linear functions with respect to $S^*$.
Note that the runtime of the mechanism is exponential\footnote{Indeed, in the second step of the mechanism, we can use \submax\ to compute an approximate solution, which suffices for our purpose.
Step~4 can be done easily by a demand query. Hence, if we are given an access to an oracle which, for any subset $X$ of items, gives a linear function $f$ with $f(X)=v(X)$, then the mechanism can be implemented in polynomial time.}.

Consider the threshold $t$, subset $S^*$, and additive function $f$ defined in the \XOSsample. We have the following observation.

\begin{claim}\label{lem:subset}
For any $S \subseteq S^*$, $f(S) - t\cdot c(S) \geq 0$.
\end{claim}

\begin{proof}
Suppose by a contradiction that there exists a subset $S \subseteq S^*$ such that $f(S) - t\cdot c(S) < 0$. Let $S' = S^* \setminus S$.
Since $f$ is an additive function, we have $c(S') + c(S) = c(S^*)$ and $f(S') + f(S) = f(S' \cup S) = f(S^*) = v(S^*)$.
Thus,
\begin{eqnarray*}
  v(S') - t\cdot c(S') & \geq & f(S') - t\cdot c(S') \\
  & = & v(S^*) - t\cdot c(S^*) - \big(f(S) - t\cdot c(S)\big) \\
  & > & v(S^*) - t\cdot c(S^*),
\end{eqnarray*}
which contradicts the definition of $S^*$.
\end{proof}

The following claim is critical for truthfulness.

\begin{claim} \label{lem:SameS}
If any item $j\in S^*$ reports a different cost $b(j) < c(j)$, then set $S^*$ remains the same.
\end{claim}
\begin{proof}
Let $b$ be the bid vector where $j$ reports $b(j)$ and others remain unchanged. First we notice that for any set $S$ with $j \in S$,
$\big(v(S) - t\cdot b(S)\big) - \big(v(S) - t\cdot c(S)\big) = t\big(c(j) - b(j)\big)$ is a fixed positive value.
Hence,
\begin{eqnarray*}
    v(S^*) - t\cdot b(S^*) & = & v(S^*) - t\cdot c(S^*) + t\big(c(j) - b(j)\big)\\
    & \geq & v(S) - t\cdot c(S) + t\big(c(j) - b(j)\big) \\
    & = & v(S) - t\cdot b(S).
\end{eqnarray*}
Further, for any set $S$ with $j \notin S$, we have
\begin{eqnarray*}
    v(S^*) - t\cdot b(S^*) & > & v(S^*) - t\cdot c(S^*) \\
    & \geq & v(S) - t\cdot c(S) \\
    & = & v(S) - t\cdot b(S).
\end{eqnarray*}
Therefore, we may conclude that $S^* = \arg\max\limits_{S\subseteq A\setminus T}\big(v(S)-t\cdot b(S)\big).$
\end{proof}

Our main mechanism for XOS functions, denoted by \underline{\XOSmain}, is simply a uniform
distribution of the two mechanisms \litem\ and \XOSsample. We have the following result.

\begin{theorem}\label{theorem-XOS-main}
The mechanism \XOSmain\ is budget feasible, universally truthful, and provides a constant approximation ratio for XOS valuation functions.
\end{theorem}

The theorem follows from the following lemmas.

\begin{lemma}
The main mechanism \XOSmain\ is universally truthful.
\end{lemma}
\begin{proof}
Our mechanism is a combination of two mechanisms, in which \litem\
is obviously truthful. Therefore, it remains to prove that \XOSsample\
is truthful. To this end, since all items are single parameter,
it suffices to show that \XOSsample\ is monotone, that is, a winning item will still be in the winning set with a smaller bid.
Assume that item $i$ is in the winning set of \XOSsample. If $i$ decreases its bid, then
by Claim~\ref{lem:SameS} and the rule of the mechanism, $S^*$ does not change.
When the mechanism runs \AddM\ for $S^*$ with respect to additive function $f(\cdot)$,
since \AddM\ is a truthful mechanism,
$i$ will still be in the winning set when decreasing its bid.
Therefore, \XOSsample\ is monotone, and thus, truthful.
\end{proof}

\begin{lemma}
  The main mechanism \XOSmain\ is budget feasible.
\end{lemma}
\begin{proof}
  It suffices to prove that both \litem\ and \XOSsample\ are budget feasible.
  Clearly, \litem\ is budget feasible. \XOSsample\ uses a budget feasible
  mechanism \AddM\ as a final output for winning set and payments to them. Therefore,
  threshold payments in \XOSsample\ can be only smaller than those in \AddM\ providing us
  that \XOSsample\ is budget feasible as well.
\end{proof}

\begin{lemma}
The main mechanism \XOSmain\ has a constant approximation ratio.
\end{lemma}

\begin{proof}
Let $OPT$ denote the optimal solution given budget $B$, and let $k = \min_{i\in OPT}{\frac{v(OPT)}{v(i)}}$.
Thus  $v(OPT) \geq k\cdot v(i)$ for each $i \in OPT$. By Lemma~\ref{lem:probability},
we have $v(OPT \cap T) \geq \frac{k-1}{4k}v(OPT)$ with probability at least $\frac{1}{2}$.
Thus, if we denote the optimal solution of $T$ given budget $B$ by $OPT(T)$,
then we have $v(OPT(T)) \geq v(OPT \cap T) \geq \frac{k-1}{4k}v(OPT)$ with
probability at least $\frac{1}{2}$, as $OPT \cap T$ is a particular solution
and $OPT_T$ is an optimal solution for set $T$ with budget constraint.


We let $OPT^*$ be the optimal solution with respect to the item set
$S^*$, additive value function $f$ and budget $B$. In the
following we will show that $f(OPT^*)$ is a good approximation to
the actual social optimum $v(OPT)$. Consider the following two
cases:
\begin{itemize}
\item $c(S^*) > B$. Given the condition, we can always find a
    subset $S' \subseteq S^*$, such that $\frac{B}{2} \leq c(S')
    \leq B$. By Claim~\ref{lem:subset}, we know $f(S') \geq
    t\cdot c(S') \geq \frac{v(OPT(T))}{8B}\cdot\frac{B}{2} \geq
    \frac{v(OPT(T))}{16}$. Then by the fact that $OPT^*$ is an optimal
    solution and $S'$ is a particular solution with budget constraint $B$, we have $f(OPT^*) \geq f(S') \geq
    \frac{v(OPT(T))}{16}\geq \frac{k-1}{64k}v(OPT)$ with probability
    at least $\frac{1}{2}$.

\item $c(S^*) \leq B$. Then $OPT^* = S^*$. Let $S' = OPT \backslash T$; thus, $c(S') \leq c(OPT) \leq B$.
    By Lemma~\ref{lem:probability}, we have $v(S') \geq \frac{k-1}{4k}v(OPT)$
    with probability at least $\frac{1}{2}$. Recall that $S^* = \arg\max\limits_{S\subseteq A\setminus
      T}(v(S)-t\cdot c(S))$. Then with probability at least $\frac{1}{2}$, we have
    \begin{eqnarray*}
      f(OPT^*) = f(S^*) & = & v(S^*) \\
      & \geq & v(S^*) - t\cdot c(S^*) \\
      & \geq & v(S') - t\cdot c(S') \\
      & \geq &  \frac{k-1}{4k}v(OPT) -\frac{v(OPT(T))}{8B}\cdot B \\
      & \geq & \frac{k-1}{4k}v(OPT) -\frac{v(OPT)}{8} \\
      & = & \frac{k-2}{8k}v(OPT).
    \end{eqnarray*}
\end{itemize}

In either case, we get $f(OPT^*) \geq
\min\left\{\frac{k-1}{64k}v(OPT),\frac{k-2}{8k}v(OPT)\right\} \ge  \frac{k-2}{64k}v(OPT)$ with probability at least $\frac{1}{2}$. Then
  in the last step of our mechanism \XOSsample, we use the output of
  \AddM$(f, S^*, B)$ as our final output. Recall that \AddM\ has
  approximation factor of at most 3 with respect to the optimal
  solution $f(OPT^*)$. Thus, the solution given by \XOSsample\ is at least $\frac{1}{3}\cdot
  f(OPT^*) \geq \frac{1}{3} \cdot \frac{1}{2} \cdot \frac{k-2}{64k}
  v(OPT) = \frac{k-2}{384k}v(OPT)$.

  On the other hand, since $k = \min_{i\in OPT}{\frac{v(OPT)}{v(i)}}$, the solution given by
   \litem\ satisfies $\max_i{v(i)} \geq \frac{1}{k}v(OPT)$.
  Combining the two mechanisms together, our main mechanism \XOSmain\ has performance at least
  \[\left(\frac{1}{2}\cdot\frac{k-2}{384k} + \frac{1}{2}\cdot\frac{1}{k}\right)v(OPT) = \frac{k+382}{768k}v(OPT) \ge \frac{1}{768}v(OPT).\]
  This completes the proof of the lemma.
\end{proof}

\subsection{Subadditive Functions}\label{sec:subadditive}

\newcommand{\SubMainm}{{\sc SA-mechanism-main-2}}
\newcommand{\nv}{{\widetilde{v}}}
\newcommand{\Family}

Here we use our results of the previous subsection for XOS functions
to design a truthful mechanism for subadditive functions.
Let $S_1,\ldots,S_N$ be a permutation of all possible subsets of $A$, where $N=|2^{A}|$ is the size of the power set $2^A$.
For each subset $S\subseteq A$, consider the following linear program, where there is a variable $\alpha_j$ associated with each subset $S_j$.
\begin{eqnarray*}
LP(S):~~~~~ &\min& \sum\limits_{j=1}^N \alpha_{j}\cdot v(S_j)\\
&s.t.& \alpha_j\ge 0,\quad 1\le j\le N \\
&& \sum\limits_{j:\ i\in S_j} \alpha_j \ge 1,\quad \forall\ i\in S
\end{eqnarray*}
In the above linear program, the minimum is taken over all possible non-negative values of $\alpha=(\alpha_1,\ldots,\alpha_N)$.
If we consider each $\alpha_j$ as the fraction of the cover by subset $S_j$, the last constraint implies that all items in $S$ are fractionally covered.
Hence, LP$(S)$ describes a linear program for the set cover of $S$.
For any subadditive function $v(\cdot)$, it can be seen that the
value of the optimal integral solution to the above LP$(S)$ is
always $v(S)$. Indeed, one has $S\subseteq \bigcup_{j:\ \alpha_j\ge
1} S_j$ and $\sum_{j}\alpha_{j}\cdot v(S_j)\ge\sum_{j:\ \alpha_j\ge
1}v(S_j)\ge v\big(\bigcup_{j:\ \alpha_j\ge 1} S_j\big)\ge v(S)$.

Let $\nv(S)$ be the value of the optimal fractional solution of
LP$(S)$, and $\I(S)=\frac{v(S)}{\nv(S)}$ be the integrality gap of
LP$(S)$. Let $\I=\max_{S\subseteq A}\I(S)$; the integrality gap $\I$
gives a worst-case upper bound on the integrality gap of all
subsets. Hence, we have $\frac{v(S)}{\I}\le\nv(S)\le v(S)$ for any
$S\subseteq A$. The classic Bondareva-Shapley
Theorem~\cite{bondareva63,shapley67} says that the integrality gap
$\I(S)$ is one (i.e., $v(S)$ is also an optimal fractional solution
to the LP) if and only if $v(\cdot)$ is an XOS function.

\begin{lemma}
$\nv(\cdot)$ is an XOS function.
\end{lemma}
\begin{proof}
For any subset $S\subseteq A$, consider any non-negative vector $\gamma=(\gamma_1,\ldots,\gamma_N)\ge 0$ that satisfies $\sum_{\substack{j:~ i\in S_j}} \gamma_j \ge 1$ for any $i\in S$.
Then, we have
\begin{eqnarray*}
\sum_{j=1}^{N}\gamma_j\cdot \nv(S_j) &=& \sum_{j=1}^{N}\gamma_j\cdot \min_{\beta_{j,\cdot}\ge 0} \left(\sum_{k=1}^{N}\beta_{j,k}\cdot v(S_k) ~\middle|~ \forall\ i\in S_j, \sum_{\substack{k:~ i\in S_k}} \beta_{j,k} \ge 1 \right) \\
&=& \min_{\beta\ge 0} \left(\sum_{j=1}^N \gamma_j \sum_{k=1}^N \beta_{j,k} \cdot v(S_k) ~\middle|~ \forall\ j, \forall\ i\in S_j, \sum_{\substack{k:~ i\in S_k}} \beta_{j,k} \ge 1 \right) \\
&=& \min_{\beta\ge 0} \left(\sum_{k=1}^N \bigg(\sum_{j=1}^N \gamma_j \beta_{j,k} \bigg)\cdot v(S_k) ~\middle|~ \forall\ j, \forall\ i\in S_j, \sum_{\substack{k:~ i\in S_k}} \beta_{j,k} \ge 1 \right) \\
&\ge& \min_{\alpha\ge 0} \left(\sum_{k=1}^{N}\alpha_k\cdot v(S_k) ~\middle|~ \forall\ i\in S, \sum_{\substack{k:~ i\in S_k}} \alpha_k \ge 1 \right) \\
&=& \nv(S)
\end{eqnarray*}
The inequality above follows from the fact that for any $i\in S$,
\[\sum_{k:~i\in S_k} \sum_{j} \gamma_j \beta_{j,k} = \sum_{j}\gamma_j \sum_{k:~i\in S_k} \beta_{j,k} \ge \sum_{j}\gamma_j  \ge \sum_{j:~i\in S_j}\gamma_j \ge 1.\]
Hence, $\nv(\cdot)$ is fractionally subadditive, which is equivalent to XOS.
\end{proof}

We are now ready to present our mechanism for subadditive functions.

\begin{center}
\small{}\tt{} \fbox{
\parbox{6.0in}{\hspace{0.05in} \\
[-0.05in] \underline{\SubMainm}
\begin{enumerate}
\item For each subset $S\subseteq A$, compute $\nv(S)$.
\item Run \XOSmain\ for the instance with respect to XOS function $\nv(\cdot)$.
\item Output the result of \XOSmain. \\[-.2in]
\end{enumerate}
} }
\end{center}

\begin{theorem}
The mechanism \SubMainm\ is budget feasible, universally truthful, and provides an approximation ratio of $O(\I)$ for subadditive functions,
where recall that $\I$ is the largest integrality gap of LP$(S)$ for all subsets.
\end{theorem}
\begin{proof}
Note that the valuations $v(\cdot)$ are public knowledge; thus computing $\nv(\cdot)$ and run \XOSmain\ with respect to $\nv(\cdot)$ do not affect truthfulness.
The claim then follows from Theorem~\ref{theorem-XOS-main} and the fact that $\frac{v(S)}{\I}\le\nv(S)\le v(S)$ for any $S\subseteq A$
(i.e., using $\nv(\cdot)$ instead of $v(\cdot)$ we only lose a factor of $\I$ in the approximation ratio).
\end{proof}

%
%

\section{Extensions}\label{section-extension}

In the current section we consider two extensions for valuation
functions where the mechanisms described before still can be
applied.

\subsection{Non-Monotone Functions}

In general, $v(\cdot)$ can be a non-monotone subadditive (or XOS) function, e.g., the cut function studied in~\cite{DPS11}. That is, for any $S\subset T\subseteq A$, it is not necessarily that $v(S)\le v(T)$.
We next describe how to apply our mechanisms to non-monotone functions.

For any subset $S\subseteq A$, define
\[
\hat{v}(S)=\max\limits_{T\subseteq S}v(T).
\]
Clearly, $\hat{v}(\cdot)$ is monotone and inherits the classification of $v(\cdot)$; that is, if $v(\cdot)$ is subadditive (or XOS), so does $\hat{v}$.
Note that given a demand oracle, $\hat{v}(\cdot)$ can be computed easily.
Then we can apply our mechanisms to $\hat{v}(\cdot)$ directly.
Further, we have the following observations.
\begin{itemize}
\item For any subset $S\subseteq A$, let $OPT(S)$ be an optimal solution of $v(\cdot)$ on $S$. Then $OPT(S)$ is an optimal solution of $\hat{v}(\cdot)$ on $S$ as well. In particular, this implies that $\hat{v}(\cdot)$ and $v(\cdot)$ will have the same optimal value on the whole set and testing set in random sampling.
\item In mechanism \XOSsample, the computed $S^*\in \arg\max_{S\subseteq A\setminus T}\big\{v(S)-t\cdot c(S)\big\}$ is an optimal solution for $\max_{S\subseteq A\setminus T}\big\{\hat{v}(S)-t\cdot c(S)\big\}$ as well, i.e., $\hat{v}(S^*)=v(S^*)$. Let $\hat{f}$ be the linear function with $\hat{f}(S^*)=\hat{v}(S^*)=v(S^*)$ computed for $\hat{v}(\cdot)$. Note that in $\hat{f}(S^*)$, each item $i\in S^*$ has a non-negative contribution to the value of $\hat{f}(S^*)$ (otherwise, $S^*$ will not be an optimal solution). Hence, we can run \AddM\ for additive function $\hat{f}(\cdot )$ on $S^*$, which yields the desired result.


\end{itemize}
Therefore, all our mechanisms described above continue to work for non-monotone functions with the same approximation ratios.

\subsection{Relaxed Subadditive Functions}

A valuation function $v(\cdot)$ is called $\calK$-subadditive if for any disjoint subsets $S_1, S_2, \ldots,S_\ell\subseteq A$ of items,
$v(S_1\cup S_2\cup\cdots\cup S_\ell) \le \calK\cdot \big(v(S_1) + v(S_2) + \cdots + v(S_\ell)\big)$.
Note that a function is subadditive in the usual sense if and only if it is 1-subadditive.

For this case we may consider another valuation function: For any $S\subseteq A$, define
\[
\check{v}(S)=\min\Big\{v(S_1) + v(S_2) + \cdots + v(S_\ell)~\big|~S_1,\dots,S_\ell \ \textup{is a partition of $S$}\Big\}.
\]
Note that $\check{v}(\cdot)$ approximates $v(\cdot)$ within a factor
of $\calK$, that is, $\calK\cdot\check{v}(S) \ge v(S)\ge
\check{v}(S).$ It can be seen that this new function is subadditive.
Thus, by applying our mechanisms to $\check{v}(\cdot)$, we lose an
extra factor of $\calK$ in the approximation ratio with respect to
the optimal solution of $v(\cdot)$. In particular, when $\calK$ is a
constant, our mechanisms will have the same order of approximation
ratio.

\section{Special Examples}\label{sec:special}

In this section we consider a few concrete examples of XOS and
subadditive functions. The main purpose of which is to illustrate
how general scheme works in particular settings, and give certain
evidence that the general approach is natural and can be efficiently
implemented in certain circumstances.

\subsection{Matching}

In the instance of matching, we are given a (bipartite) graph $G=(U,E)$, where each edge $e\in E$ corresponds to an agent with a value $v(e)$ and a privately known cost $c(e)$. For any subset of edges $S\subseteq E$, its value $v(S)$ is defined to be the total value of the largest matching induced by the edge set $S$. It is well known that matching is not submodular (e.g., edge $(u_3,u_4)$ contributes one to the set $\{(u_1,u_2),(u_2,u_3)\}$ but contributes zero to $\{(u_2,u_3)\}$ when all edges have unit value). However, matching is in the class of XOS functions; hence, our mechanism \XOSmain\ gives a constant approximation.

We argue that our mechanism \XOSsample, and thus, \XOSmain, can be implemented in polynomial time for matching.
In random sampling, computing an optimal solution for the testing group $T$ is equivalent to solving a maximum weighted matching problem with a budget constraint,
which admits a polynomial-time approximation scheme~\cite{BBG11}.
(Note that similar to \submain\ and its subroutine \submax, it is not necessary to compute an optimal solution for the testing group; any constant approximation suffices to our mechanisms.)
The set $S^*$ can be computed according to the following simple subroutine: for each remaining edge $e\in E\setminus T$
we set the new value $w(e)=v(e)-t\cdot c(e)$ and then compute a maximal matching with respect to $w$ in the induced subgraph $E\setminus T$; this forms the set $S^*$. Finally, for this given selected matching $S^*$, its valuation is already additive with respect to its members. Therefore, the implementation can be done in polynomial time.
This gives the following claim.

\begin{proposition}
\XOSmain\ is a constant approximation mechanism for matching and can be implemented in polynomial time.
\end{proposition}

\subsection{Clique}

Given a graph $G=(U,E)$, each vertex $i$ is an agent with a value $v(i)$ and a privately known cost $c(i)$. For a given subset of vertices $S\subseteq U$, its value $v(S)$ is defined to be the value of the largest weighted clique in $S$, i.e., $v(S)=\max\big\{\sum_{i\in T}v(i)~|~T\subseteq S \ \textup{is a clique}\big\}$. Note that clique is not submodular as well (e.g., consider a graph with vertices $\{i_1,i_2,i_3,i_4\}$ of unit value each and edges $\{(i_1,i_2),(i_1,i_4),(i_2,i_3),(i_2,i_4)\}$; the contribution of $i_1$ to $\{i_2,i_3\}$ is zero, but to $\{i_2,i_3,i_4\}$ is one). Further, it can be seen that clique is an XOS function (this follows simply from the definition of $v(S)$).
Hence, given a demand oracle, \submax\ computes a subset whose value is a constant approximation to the optimal solution.
In addition, the set $S^*$ can be computed by a single demand query; and the linear function for $v(S^*)$ can be found easily by demand queries for elements in $S^*$ one by one. Hence, we have the following claim.

\begin{proposition}
\XOSmain\ plus subroutine \submax\ gives a constant approximation
mechanism for clique and can be implemented in polynomial time given
a demand oracle\footnote{Our result does not contradict to the
hardness of max-clique approximation. Indeed, given the access to
powerful oracles, even a simple value query of the whole set will
give the value of an optimal clique solution.}.
\end{proposition}

\subsection{From Cost Minimization to Valuation Maximization}

In this section we will consider our motivating example discussed at the beginning of the Introduction.
A company serves a set $A$ of agents (or jobs). Let $\calR$ denote the set of solution spaces, i.e., all possible ways to serve the agents.
For any $r\in \calR$, there is a cost function $f_r(\cdot)$ which gives the running cost to serve different subsets of agents by solution $r$.
For any subset $S\subseteq A$, we want to spend as little cost as possible. Hence, the cost to the company is given by
\[c(S) = \min\big\{f_r(S)~|~r\in \calR\big\}.\]
The model includes a number of well-studied problems. For example, in job scheduling, $\calR$ corresponds to all possible assignments between jobs and machines;
in facility location, $\calR$ includes all combinations of facilities to open; and in congestion games, $\calR$ gives all possible assignments between agents and resources.
For each of these minimization problems, the cost of the whole set $c(A)$ gives the value of an optimal solution.

Assume now the company would like to remove some of the agents from $A$ by paying them a certain amount of subsidy. Our objective is save the running cost as much as possible. Hence,
for any $S\subseteq A$, we can define
\[v(S)=c(A)-c(A\setminus S).\]
That is, the valuation of $S$ is equal to the difference of the costs between the whole set $A$ and the remaining set $A\setminus S$.


\begin{proposition}
If $c(\cdot)$ is a supermodular function, then $v(\cdot)$ is a submodular function.
\end{proposition}

The above claim applies to, e.g., congestion games when the latency functions are polynomials
with positive coefficients (thus the cost function is supermodular);
hence, our mechanisms can be applied, as well as mechanisms specifically designed for submodular valuation~\cite{PS10,CGL11}.
When $c(\cdot)$ is a submodular function, it can be seen that $v(\cdot)$ is neither submodular nor subadditive.
However, in some examples (e.g., in job scheduling, the cost is a square root function with respect to the total length of
scheduled jobs), we can show that $v(\cdot)$ is a relaxed subadditive function, and our mechanisms therefore can be
applied.

Our model links cost minimization problems with budget feasible mechanism design.
For a variety of problems, e.g., facility location, and job scheduling with various minimization objectives,
the valuation function $v(\cdot)$ defined above may not necessarily fall into the class of (relaxed) subadditive
functions. The design of budget feasible mechanisms with small approximations for these problem
is an intriguing question, which we leave for future work.

\appendix

\end{document}